\documentclass[12pt]{amsart}

\usepackage[T1]{fontenc}
\usepackage[hmargin={25mm,25mm}]{geometry}

\usepackage{color}
\usepackage{amsmath,amssymb,amsthm}
\usepackage[bookmarks=false,colorlinks=true]{hyperref}
\usepackage{latexsym,bbm,bm}
\newcommand{\sfrac}[2]{{\textstyle\frac{#1}{#2}}}
\newcommand{\ba}{\begin{array}}
\newcommand{\ea}{\end{array}}
\newcommand{\be}{\begin{equation}}
\newcommand{\ee}{\end{equation}}
\newcommand{\bea}{\begin{eqnarray}}
\newcommand{\eea}{\end{eqnarray}}

\newcommand{\Hgencoul}{{\mathcal{H}}_\gamma}

\newcommand{\cal}{\mathcal}

\newcommand{\A}{\mathcal A}
\newcommand{\mL}{\mathcal L}

\def \C {\mathbb{C}}
\def \R {\mathbb{R}}

\def \N {\mathbb{N}}

\def \H_g so(N){H_g^{so(N)}}
\def \Hggl(N){H_g^{gl(N)}}
\newcommand{\nad}[2]{\genfrac{}{}{0pt}{}{#1}{#2}}

\def\beq#1#2\eeq{%
        \begin{equation}%
        \label{#1}%
            #2%
        \end{equation}%
    }

\newtheorem{thm}{Theorem}[section]
\newtheorem*{thm*}{Theorem}
\newtheorem{cor}[thm]{Corollary}
\newtheorem{lem}[thm]{Lemma}
\newtheorem{prop}[thm]{Proposition}
\newtheorem{rem}[thm]{Remark}

\title{Algebra of Dunkl Laplace--Runge--Lenz vector}

\author{Misha Feigin}
\address{School of Mathematics and Statistics, University of Glasgow, 15 University Gardens, Glasgow G12 8QW, UK}
\email{misha.feigin@glasgow.ac.uk}

\author{Tigran Hakobyan}
\address{Yerevan State University, 1 Alex Manoogian St., Yerevan, 0025, Armenia}
\email{tigran.hakobyan@ysu.am}

\begin{document}

\begin{abstract}
We consider  Dunkl version of Laplace--Runge--Lenz vector associated with a finite Coxeter group $W$ acting geometrically in $\mathbb R^N$ with multiplicity function  $g$. This vector 
generalizes the usual Laplace--Runge--Lenz vector and its components commute  with Dunkl--Coulomb Hamiltonian given as Dunkl Laplacian with additional Coulomb potential $\gamma/r$. We study resulting symmetry algebra $R_{g, \gamma}(W)$ and show that it has Poincar\'e--Birkhoff--Witt property. In the absence of Coulomb potential  this symmetry algebra $R_{g,0}(W)$ is a subalgebra of the rational Cherednik algebra $H_g(W)$. We show that a central quotient of the algebra  $R_{g, \gamma}(W)$ is a quadratic algebras isomorphic to a central quotient of the corresponding Dunkl angular momenta algebra $H_g^{so(N+1)}(W)$. This gives interpretation of the algebra  $H_g^{so(N+1)}(W)$ as the hidden symmetry algebra of Dunkl--Coulomb problem in $\R^N$. By specialising $R_{g, \gamma}(W)$ to $g=0$ we recover a quotient of the universal enveloping algebra
$U(so(N+1))$ as the hidden symmetry algebra of Coulomb problem in ${\mathbb R}^N$.
We also apply Dunkl  Laplace--Runge--Lenz vector  to establish maximal superintegrability of generalised Calogero--Moser systems.
\end{abstract}

\maketitle

\section{Introduction}

Discovery of classical version of Laplace--Runge--Lenz (LRL) vector goes back to 1710 letter of Johann Bernoulli and an earlier letter by J. Hermann the same year (see \cite{Gold} and references therein). Quantum version of the LRL vector was used by Pauli in 1926 for the derivation of the spectrum of the hydrogen atom \cite{Pauli}. Pauli also found out relations between components of LRL vector, the Hamiltonian of hydrogen atom (equivalently, Coulomb problem Hamiltonian), and angular momenta operators. A. Hulth\'en (with  acknowledgment to O. Klein) pointed out a connection to orthogonal
(or rather Lorentz) group in four-dimensional space in 1933 \cite{Hul}. Indeed, it is clear from these relations that conserved quantities at a fixed level of energy satisfy $so(4)$ relations, which is also called the hidden symmetry algebra of Coulomb problem as it extends more straightforward $so(3)$ symmetry algebra generated by angular momenta.  In the case of Coulomb problem in $N$-dimensionl space the  hidden symmetry algebra is $so(N+1)$ \cite{revai65}. Another explanation of the hidden rotational symmetry of Coulomb problem based  on momentum space representation of wave functions was presented in \cite{Fock}.



Dunkl operators arose in the theory of generalised harmonic polynomials assocaited with a root systems $\mathcal R$ of a finite reflection groups $W$ and $W$-invariant function $g: {\mathcal R} \to \C$ \cite{dunkl}. These operators give a deformation of partial derivatives with non-local additional terms which vanish if $g=0$. Dunkl operators pairwise commute and the sum of squares of these operators corresponding to an orthonormal basis gives Dunkl Laplacian which is a deformation of the usual Laplace operator \cite{dunkl0, dunkl-book}.

Dunkl operators are also key ingredients of the rational Cherednik algebras \cite{EG}, and they are closely related to Calogero--Moser integrable systems.
The  Calogero--Moser system describes interacting
particles on a line with a pairwise inverse  square distance potential
\cite{calogero0}. Liouville integrability of the classical system was established by the Lax method in \cite{moser}.
 Calogero--Moser system has integrable generalizations
related to root systems $\mathcal R$ of finite reflection groups $W$  \cite{rev-olsh}. Integrability of these systems can be established with the help of Dunkl operators.
Indeed, the restriction of Dunkl Laplacian to the space of $W$-invariant functions gives generalised Calogero--Moser Hamiltonian associated with $\mathcal R$. Moreover, integrals of motions can also be expressed as restrictions of $W$-invariant combinations of Dunkl operators \cite{Heckman} (see also \cite{poly92} for the usual Calogero--Moser system in harmonic confinement).


In addition to Liouville integrability the Calogero--Moser Hamiltonian has extra integrals of motion
ensuring its maximal superintegrability. This was established for the classical Calogero--Moser system in  \cite{woj83}. 
Maximal superintegrability of a classical integrable system leads to the property that compact trajectories are closed. In the  quantum case superintegrability corresponds to degeneration of spectrum. 
Maximal superintegrability of quantum Calogero--Moser system was shown in \cite{K}, where additional integrals were constructed with the help of Dunkl operators (see also \cite{gonera} for another proof).  Superintegrability for the quantum system was claimed in \cite{CFS} for any root system and additional integrals were given. Algebraic structures formed by  Calogero--Moser operators together with their integrals were investigated in \cite{K}, \cite{CFS} (see also \cite{HUW} where a closely related quadratic algebra was considered). 
Superintegrability of spin Calogero--Moser systems was investigated in \cite{R1}, \cite{R2}.

A related algebraic structure is the Dunkl angular momenta algebra $H_g^{so(N)}(W)$ \cite{fh}. This algebra is generated by Dunkl angular momenta operators and the group $W$. It can be thought of as a flat deformation of skew product of a quotient of the universal enveloping algebra $U(so(N))$ and $W$. The central element of this algebra
acts naturally on $W$-invariant functions as the  angular part of the corresponding Calogero--Moser Hamiltonian.
Central quotient of the algebra $H_g^{so(N)}(W)$ can be identified with the algebra of global sections of a sheaf of Cherednik algebras on a smooth quadric \cite{FT}. 
 Dunkl angular momenta also lead to various symmetries of a Dirac operator in the Clifford extension of the rational Cherednik algebra studied in \cite{dirac}.

Dunkl--Coulomb Hamiltonian $H_{g, \gamma}$ is the non-local operator given as Dunkl Laplacian with additional Coulomb potential $\gamma/r$, it depends on the coupling parameter $g$ and the Coulomb force parameter $\gamma$. Such two-dimensional Hamiltonian  was considered in \cite{vinet} for the root system $A_1\times A_1$, where a version of LRL vector for the corresponding Hamiltonian was presented. In the $N$-dimensional case and the root system $A_{N-1}$ the Dunkl LRL vector was  introduced in  \cite{Runge}  (see also \cite{CalCoul} for the initial attempt and for 
discussion of superintegrability of classical generalised Calogero--Moser Hamiltonians with Coulomb potential).   It commutes with the Dunkl--Coulomb Hamiltonian $H_{g, \gamma}$. If the coupling parameter $g=0$ then the Dunkl LRL vector reduces to the usual LRL vector. On the other hand if $\gamma=0$ then one gets a version of the LRL vector for the  Dunkl Laplacian. Dunkl LRL vector for type $A$ Dunkl--Coulomb problem on the sphere was considered in \cite{francisco}.
For 
certain integrable perturbations, which break $H_g^{so(N)}(W)$ symmetry  along a
particular direction, the corresponding component of the LRL vector is 
preserved \cite{HNer}.

In the present paper we introduce Dunkl LRL vector for the Dunkl--Coulomb Hamiltonian $H_{g, \gamma}$  associated with the
root system $\mathcal  R$ of an arbitrary finite Coxeter group $W$. Components of this vector  commute with the Dunkl--Coulomb Hamiltonian. Another set of operators commuting with the Hamiltonian is given by Dunkl angular momenta. 
This gives a new way to prove  maximal superintegrability of quantum generalised Calogero--Moser system related with a root system $\mathcal R$. Furthermore, it leads to 
additional quantum integrals 
of  generalised Calogero--Moser systems which do not have full Coxeter symmetry and are related to special representations of rational Cherednik algebras \cite{Fsel}, \cite{CFV}, \cite{SV}. Components of Dunkl LRL vector, Dunkl angular momenta, Dunkl--Coulomb Hamiltonian and elements of $W$ generate algebra $R_{g,\gamma}(W)$ which may be thought of as the symmetry algebra of the Hamiltonian $H_{g, \gamma}$. It is the main object of the present work.

After reviewing background information in Section \ref{sec2} we establish relations 
involving Dunkl LRL vector
in Section \ref{section-AR}. In Section \ref{section-BR} we establish defining relations of the algebra $R_{g, \gamma}(W)$ and find its basis.   Then we consider a central quotient of $R_{g,\gamma}(W)$ and show that it is isomorphic to a central quotient of the Dunkl angular momenta algebra $H_g^{so(N+1)}(W)$ where the group $W$ acts naturally on the first $N$ components of vectors in $\C^{N+1}$. These central quotients are (non-homogeoneous) quadratic PBW algebras in the sense of \cite{BG}. 
This isomorphism gives an interpretation of the Dunkl angular momenta algebra $H_g^{so(N+1)}(W)$ as the hidden symmetry algebra of the Dunkl--Coulomb problem in $\R^N$ similarly to (the quotient of the universal enveloping algebra of)  $so(N+1)$ being the hidden symmetry algebra of Coulomb problem in $\R^N$. The latter property can also be deduced from our considerations by specialising the algebra $R_{g, \gamma}(W)$ at $g=0$ so that Dunkl--Coulomb Hamiltonian $H_{g, \gamma}$ takes the form of the usual Coulomb Hamiltonian in $\R^N$. At this specialisation the central quotient of $R_{0, \gamma}(W)$ becomes isomorphic to the skew product of a quotient of the universal enveloping algebra $U(so(N+1))$ and Coxeter group $W$.

Of particular interest is the case $\gamma=0$ when Coulomb potential is absent since the algebra $R_{g, 0}(W)$ is a subalgebra of the corresponding rational Cherednik algebra $H_g(W)$. 
Even though $H_g^{so(N+1)}(W)$ is naturally a subalgebra of the rational Cherednik algebra with group $W$ acting in $(N+1)$-dimensional vector space  its central quotient appears to be isomorphic to a central quotient of the subalgebra $R_{g, 0}(W)$ in the rational Cherednik algebra $H_g(W)$ with group $W$ acting in the space $\C^N$. 

In Section \ref{superint} we apply   developed Dunkl LRL vector and algebra $R_{g,\gamma}(W)$, and we establish maximal superintegrability of  Calogero--Moser systems for any root system $\mathcal R$ with additional (possibly, vanishing) Coulomb potential. We also establish maximal superintegrability of generalisations of such systems which do not have full Coxeter symmetry. 

\section{Dunkl operators and their properties}
\label{sec2}

Let ${\cal R}$ be a Coxeter root system in $N$-dimensional Euclidean space $\R^N$ \cite{Hum}.
Let $V\cong \C^N$ be the complexification of this vector space 
 with the bilinear inner product denoted by $(\cdot, \cdot)$. The corresponding finite Coxeter group $W$ is generated by orthogonal reflections
\begin{equation*}
s_\alpha (x)= x- \frac{2(\alpha, x)}{(\alpha,\alpha)}\alpha,
\quad
\alpha\in \mathcal{R}, \,  x \in V.
\end{equation*}
Root system $\cal R$ can be represented as disjoint union ${\cal R}={\cal R}_+ \cup {\cal R}_-$, where ${\cal R}_+$ is a system of non-collinear positive roots and ${\cal R}_- = -{\cal R}_+$. Root system $\cal R$ satisfies invariance $s_\alpha {\cal R} = {\cal R}$   $\forall \alpha \in \cal R$.

Consider a  multiplicity function $g\colon \cal R \to \C$. Let $g_\alpha = g(\alpha)$ for $\alpha \in \cal R$. We assume that $g$ is $W$-invariant, that is
\begin{equation*}
\label{g-alpha}
g_{w(\alpha)}=g_\alpha
\end{equation*}
for any $w \in W$, $\alpha \in \cal R$.

The Dunkl operator $\nabla_\xi$ for any  $\xi\in V$ is defined by  \cite{dunkl}
\begin{equation*}
\label{dunkl}
\nabla_\xi = \partial_\xi - \sum_{\alpha \in {\mathcal R}_+} \frac{g_\alpha (\alpha,\xi)}{(\alpha,x)}s_\alpha,
\end{equation*}
where  $\partial_\xi = (\xi, \partial)$ is partial derivative in the direction $\xi$ and reflections $s_\alpha$ act on functions $\psi(x)$ in a standard way,
\begin{equation*}
\label{s-psi}
s_\alpha \psi(x)= \psi(s_\alpha(x) ), \quad x\in V.
\end{equation*}
The Dunkl operators satisfy commutativity  \cite{dunkl}  $[\nabla_\xi, \nabla_\eta]=0$, and their commutators with linear functions produce
nonlocal exchange operators given as follows:
\be
\label{com-xieta}
\big[\nabla_\xi, (x, \eta)\big] 
 = (\xi, \eta)
+ \sum_{\alpha\in {\mathcal R}_+} \frac{2 g_\alpha (\alpha, \xi)(\alpha, \eta)}{(\alpha,\alpha)} s_\alpha,
\ee
where $\xi, \eta \in V$.
These relations reduce to the Weyl algebra commutation relations in the zero coupling limit, $g_\alpha=0$ for all $\alpha \in \cal R$.
For general coupling values they lead to rational Cherednik algebra $H_g(W)$ which is a deformation of the skew product of Weyl algebra  with  the  Coxeter group $W$ \cite{EG}.

It is also convenient to consider $N$ Dunkl operators corresponding to an orthonormal  basis $e_1, \ldots, e_N$ in $V$. For $x\in V$ and $\alpha \in \mathcal R$ we let
$$
x = \sum_{i=1}^N x_i e_i,  \quad \alpha = \sum_{i=1}^N \alpha_i e_i,
$$
where $x_i, \alpha_i \in \C$, 
and we denote $x_\alpha=(x,\alpha) = \sum_{i=1}^N \alpha_i x_i $.
We have
\begin{equation*}
\nabla_i =\nabla_{e_i}=\partial_i - \sum_{\alpha \in {\mathcal R}_+} \frac{g_\alpha \alpha_i}{x_\alpha}s_\alpha,
\end{equation*}
where
$
\partial_i = \frac{\partial}{\partial x_i}$, $i=1, \ldots, N$. 
In these notations, the commutation relations \eqref{com-xieta} take the form
 \be
\label{com-nn}
[\nabla_i,x_j]= S_{ij},
\ee
where
\be
\label{Sij}
S_{ij} =\delta_{ij} + \sum_{\alpha\in {\mathcal R}_+} \frac{2 g_\alpha \alpha_i\alpha_j}{(\alpha, \alpha)} s_\alpha,
\ee
$i, j=1, \ldots, N$, and $\delta_{ij}$ is the Kronecker symbol.  
Note that relations \eqref{com-nn} remain invariant  under the formal
Hermitian conjugation,
\be
\label{herm}
\nabla_i^+=-\nabla_i,
\qquad
x_i^+=x_i,
\qquad
S_{ij}^+=S_{ij}.
\ee

\begin{lem}
The commutators of the elements $S_{ij}$ with coordinates $x_k$  satisfy relations
\be
\label{comSx}
[S_{ij},  x_k] = [S_{kj}, x_i]
\ee
for all $i,j,k=1,\ldots, N$.
\end{lem}
\begin{proof}
For any $\alpha \in \mathcal R$ we have
\begin{equation}
\label{sij-mat}
s_\alpha (x_k)=x_k -\frac{2x_\alpha \alpha_k}{(\alpha, \alpha)},
\end{equation}
and therefore
\begin{equation*}
[x_k, s_\alpha]=\frac{2x_\alpha\alpha_k}{(\alpha, \alpha)}s_\alpha.
\end{equation*}
Hence we obtain from \eqref{Sij} that
\be
\label{xk-Sij}
[x_k,S_{ij}]=\sum_{\alpha\in \mathcal{R}_+}\frac{4g_\alpha\alpha_i\alpha_j\alpha_k }{(\alpha, \alpha)^2} x_\alpha s_\alpha.
\ee
The right-hand side of the equality \eqref{xk-Sij} is symmetric in the three indecies $i$, $j$ and $k$, which implies the required relation \eqref{comSx}.
\end{proof}

Let us consider the following element $S$ of the Coxeter group algebra $\C W$:
\be
\label{Sinv}
S=-\sum_{\alpha \in {\mathcal R}_+} g_\alpha s_\alpha.
\ee
\begin{lem}
The element $S$ satisfies the following relations:
\be
\label{Scomsa}
[S, s_\alpha]=0
\ee
for any $\alpha \in \cal R$, and
\be
\label{Siisum}
\sum_{i=1}^N S_{ii}=N-2S.
\ee
\end{lem}
\begin{proof}
The relation \eqref{Scomsa} follows from $W$-invariance of the multiplicity function \eqref{g-alpha} and the property
\begin{equation*}
w s_\alpha w^{-1}=s_{w(\alpha)},
\end{equation*}
which is valid for any $w\in W$. The relation  \eqref{Siisum} follows directly from \eqref{Sij}.
\end{proof}

The relations \eqref{sij-mat} and  \eqref{Sinv}  imply
\be
\label{x-pi}
(x, \nabla)  = \sum_{k=1}^N x_k \nabla_k =  r\partial_r +S,
\qquad
(\nabla,  x) = \sum_{k=1}^N  \nabla_k x_k =  r\partial_r - S + N,
\ee
where $r=\big(\sum_{k =1}^N x_k^2\big)^{1/2}$ and
 $\partial_r= \frac{1}{r} \sum_{k=1}^N x_k \partial_k$.

\begin{lem}\label{lem1}
The following relations take place for any $i=1,\ldots, N$:
\begin{subequations}
\be
\label{xjSij-a}
\sum_{j=1}^N x_jS_{ij}= x_i + [S, x_i],
\qquad
\sum_{j=1}^N S_{ij}x_j= x_i -  [S, x_i];
\ee
\be
\label{xjSij-b}
\sum_{j=1}^N \nabla_jS_{ij}= \nabla_i + [S, \nabla_i],
\qquad
\sum_{j=1}^N S_{ij}\nabla_j= \nabla_i -  [S, \nabla_i].
\ee
\end{subequations}
\end{lem}
\begin{proof}
It follows from \eqref{Sij} that
$$
\sum_{j=1}^N x_j S_{ij} = x_i + \sum_{\alpha\in \cal R_+} \frac{2 g_\alpha x_\alpha \alpha_i}{(\alpha,\alpha)}s_\alpha.
$$
We also have
$$
[g_\alpha s_\alpha, x_i] = -\frac{2 g_\alpha x_\alpha \alpha_i}{(\alpha,\alpha)} s_\alpha,
$$
which implies the first relation in \eqref{xjSij-a}. Other relations can be checked similarly.
\end{proof}
As a corollary of Lemma \ref{lem1} we get the following statement.
\begin{lem}\label{lem1.5}
The following (anti-)commutation relations take place for any $i=1,\ldots, N$:
\begin{subequations}
\be
\label{xjSij-a2}
\sum_{j=1}^N \{x_j, S_{ij}\}=2x_i,
\qquad
\sum_{j=1}^N [x_j, S_{ij}]=2[S, x_i];
\ee
\be
\label{xjSij-b2}
\sum_{j=1}^N \{\nabla_j, S_{ij}\}=2\nabla_i,
\qquad
\sum_{j=1}^N [\nabla_j, S_{ij}]=2[S, \nabla_i].
\ee
\end{subequations}
\end{lem}

\section{Dunkl--Coulomb model}
\label{section-AR}

\subsection{Nonlocal Hamiltonian}

Let us define Dunkl--Coulomb Hamiltonian as
\be
\label{gcoul}
\Hgencoul=
\Delta
- \sum_{\alpha\in {\mathcal R}_+}\frac{g_\alpha(g_\alpha-s_\alpha)(\alpha, \alpha)}{x_\alpha^2} +  \frac{2\gamma}{r},
\ee
where $\Delta = \sum_{j=1}^N \frac{\partial^2}{\partial x_j^2}$ is the Laplace operaor and $\gamma \in \C$ is a parameter.
In the case ${\cal R} = {\cal A}_{N-1}$ the Hamiltonian $\Hgencoul$ was considered in \cite{CalCoul, Runge}, and in the case  ${\cal R} = {\cal A}_{1} \times {\cal A}_1$ it was considered in \cite{vinet} (here we rescale Hamiltonian by a factor of $-2$).  
In the zero coupling limit $g_\alpha = 0$ for all $\alpha\in \cal R$ one gets Coulomb problem. On the other hand,
in the zero charge limit $\gamma=0$ the Hamiltonian \eqref{gcoul} reduces to the  nonlocal version of generalised  Calogero--Moser Hamiltonian associated with a root system $\cal R$.
The latter Hamiltonian can also be referred  to as Dunkl Laplacian as it can be expressed as sum of squares of Dunkl operators \cite{Heckman} (see also \cite{poly92}
for type ${\cal A}_{N-1}$) which allows to represent the Hamiltonian \eqref{gcoul} as
\begin{equation*}
\Hgencoul=  \nabla^2
+ \frac{2\gamma}{r},
\end{equation*}
where $\nabla^2 =  \sum_{i=1}^N \nabla_i^2$. 
Being restricted to the wavefunctions $\psi(x)$ which are symmetric or antisymmetric with respect to the action of $W$
\begin{equation*}
\psi(s_\alpha x)=\varepsilon \psi(x), \quad \alpha \in {\cal R}, \quad  \varepsilon = \pm 1,
\end{equation*}
the nonlocal Hamiltonian reduces to the Calogero--Moser--Coulomb model \cite{CalCoul}. 
It is obtained by replacing $s_\alpha$ in \eqref{gcoul} with $\varepsilon$.

\subsection{Dunkl angular momenta}
Let us describe some symmetries of the Hamiltonian $\Hgencoul$. In the Coulomb limit $g_\alpha=0$  the
Calogero--Moser terms are absent, and the rotational $so(N)$ symmetry exists, whose
generators are given by the angular momenta.
The Calogero--Moser  terms break the rotational symmetry so that the angular momentum
is not conserved  any more. Instead, we   construct a version of angular momenta
using the Dunkl operators. Consider the Dunkl angular momentum 
\be
\label{Mij}
 L_{ij}=x_i\nabla_{j}-x_j\nabla_{i}, 
\ee
where $1\le i,j \le N$.
These operators satisfy the commutation relation of  the
$so(N)$   Lie algebra with the Kronecker delta replaced by the elements  $S_{ij}\in \C W$ \cite{fh} (see also \cite{K} for ${\cal R} = {\cal A}_{N-1}$):
\be
\label{comLL}
[L_{ij},L_{kl}]=
L_{il}S_{kj} +  L_{jk}S_{li} - L_{ik}S_{lj} -  L_{jl}S_{ki}.
\ee
Similarly, there are closed commutation relations  between Dunkl angular momenta   and Dunkl operators or coordinates.
More precisely, the following lemma holds which can be checked directly.
\begin{lem} \label{lem2}
Dunkl angular momenta $L_{ij}$ satisfy
\be
\label{comLu}
[L_{ij}, x_k]=   x_i S_{jk} - x_j S_{ik}, \qquad
[L_{ij}, \nabla_k]=  \nabla_i S_{jk} -  \nabla_j S_{ik},
\ee
where   $i,j,k = 1,\ldots, N$.
 \end{lem}
This allows us to establish some symmetries of the Hamiltonian $\Hgencoul$.

\begin{prop}\label{prop32}
 The  Hamiltonian \eqref{gcoul} preserves Dunkl angular momenta:
\begin{equation*}
[\Hgencoul,L_{ij}]=0
\end{equation*}
for any $i,j=1, \ldots, N$.
\end{prop}
\begin{proof}
By Lemma \ref{lem2} we have
$$
[L_{ij}, x_k^2] = \{[L_{ij}, x_k], x_k\} =    x_i\{S_{jk}, x_k\} -x_j \{S_{ik}, x_k\}
$$
for any $k=1, \ldots, N$. It follows now from Lemma \ref{lem1.5} that
\begin{equation*}
[ r^2, L_{ij}]=0.
\end{equation*}
 Similarly one can check that
\be
\label{com-p2L}
[\nabla^2, L_{ij}]=0
\ee
(alternatively, see \cite{feigin}), and the statement follows.
\end{proof}

The   Dunkl angular momenta algebra $H_g^{so(N)}(W)$ was defined in \cite{fh} as the algebra generated by Dunkl angular momenta $L_{ij}$ and the Coxeter group algebra $\C W$. It
 has a second order Casimir element ${\cal I} = {\cal I}_N$, which is an analogue of
the angular momentum square:
\begin{gather}
\label{Lsq}
{\cal I}  = L^2_{(N)}    -  S(S-N+2),
\end{gather}
where
$L^2_{(N)}=\sum_{i<j}^N L_{ij}^2$,
and
$$
[{\cal I}, L_{ij}]=0
$$
for any $i, j = 1,\ldots, N$.
 The element $\cal I$ represents the angular part of the nonlocal Hamiltonian $\Hgencoul$
 \cite{fh}:
\be
\label{Hzero}
\Hgencoul= \nabla^2 + \frac{2\gamma}{r} = \partial_r^2 + \frac{N-1}{r}\partial_r + \frac{2\gamma}{r}  +\frac{{\cal I}}{r^2}.
\ee

\begin{rem}
It is stated in \cite{fh}  that the centre of $H_g^{{so}(N)}(W)$ is generated by $\cal I$ and $1\in W$. If $ W$ contains element $\sigma$ such that  $\sigma(x)= -  x$ for any $x\in V$ then it is easy to see that the centre is generated by $\cal I$ and $\sigma$. The  generator $\sigma$ is missing in \cite{fh}.
\end{rem}


Let us also recall that generators $L_{ij}$ satisfy additional {\it crossing} relations \cite{fh}
\be
\label{cros}
L_{ij}(L_{kl} - S_{kl})+L_{jk}(L_{il} -  S_{il})+L_{ki}(L_{jl} -  S_{jl})=0
\ee
for any $i,j,k,l = 1, \ldots, N$.

For any $\xi=(\xi_1, \ldots, \xi_N), \eta=(\eta_1, \ldots, \eta_N) \in V$ define $L_{\xi \eta} = \sum_{i, j =1}^N \xi_i \eta_j L_{ij}$. Then one has
\begin{equation}
\label{rel3}
w L_{\xi \eta} = L_{w(\xi), w(\eta)} w
\end{equation}
for any $w \in W$.

Relations  \eqref{comLL}, \eqref{cros}, \eqref{rel3}  are defining relations of  algebra $H_g^{so(N)}(W)$ \cite{fh}.

We will also need the next lemma which generalises well-known orthogonality relation between cooordinate vector and angular momentum vector in the three-dimensional space.

\begin{lem}
\label{orthogrel}
Relations
$$
L_{ij} x_k + L_{jk} x_i + L_{ki} x_j = x_k L_{ij} + x_i L_{jk}+ x_j L_{ki}=0,
$$
$$
L_{ij} \nabla_k + L_{jk} \nabla_i + L_{ki} \nabla_j = \nabla_k L_{ij} + \nabla_i L_{jk}+ \nabla_j L_{ki}=0
$$
hold for any $i,j,k=1, \ldots, N$.
\end{lem}
\begin{proof}
It follows from the relation \eqref{com-nn}  that
$$
L_{ij} x_k + L_{jk} x_i +L_{ki} x_j = x_i L_{kj} + x_k L_{ji} +x_j L_{ik}.
$$
Therefore
$$
\{L_{ij}, x_k\} + \{L_{jk}, x_i\} + \{L_{ki}, x_j\}  = 0,
$$
and hence
$$
x_k L_{ij} + x_j L_{ki} + x_i L_{jk} = -\frac12 ([L_{ij}, x_k] + [L_{jk}, x_i] + [L_{ki}, x_j])  = 0
$$
by Lemma \ref{lem2}.

The remaining relations can be established similarly.
\end{proof}

\subsection{Dunkl Laplace-Runge-Lenz vector}

In addition to the angular momentum symmetry, the standard Coulomb model possesses a hidden symmetry known as
Runge--Lenz or Laplace vector. In the presence of extra non-local Calogero--Moser potential term this vector can  be defined by making use of Dunkl operators.
For the Calogero--Moser--Coulomb nonlocal (Dunkl--Coulomb) problem related to the root system ${\cal R} = {\cal A}_{N-1}$ this conserved quantity was introduced in \cite{Runge} while in the case of ${\mathcal R}={\cal A}_1\times {\cal A}_1$ it was introduced in \cite{vinet}.
Now we extend this construction to an arbitrary root system $\cal R$.
We define components $A_i$ of a vector $A=(A_1, \ldots, A_N)$ by the formula
\be
\label{Ai}
 A_i=-\frac12 \sum_{j=1}^N\left\{ L_{ij},\nabla_j\right\}
+\frac12[\nabla_i,S] -\frac{\gamma x_i}{r},
\ee
where $i=1, \ldots, N$. This reduces to the usual Laplace--Runge--Lenz (LRL) vector in the zero coupling limit $g_\alpha=0$. Below in this section we derive various relations which components $A_i$ satisfy.  Most of these statements can be found in \cite{Runge}, \cite{vinet} for the root systems $ {\cal A}_{N-1}, {\cal A}_1\times {\cal A}_1$ respectively.

\begin{prop}\label{prop33}
Components \eqref{Ai} of the Dunkl LRL vector  $A$   satisfy  the relation
$$
[A_i, \Hgencoul]=0
$$
 for any $i=1, \ldots, N$.
\end{prop}
\begin{proof}
Note that for any $\alpha \in \mathcal R$
\begin{equation}\label{pisquare}
\left[s_\alpha, \sum_{k=1}^N \nabla_k^2\right] = \sum_{k=1}^N \left(\nabla_k - \frac{2 \alpha_k}{(\alpha,\alpha)} (\alpha, \nabla) \right)^2s_\alpha - \sum_{k=1}^N \nabla_k^2 s_\alpha = 0.
\end{equation}
Proposition \ref{prop32} together with relation  \eqref{pisquare} implies that
\begin{equation}
\label{final}
\Big[ A_i\,, \, \Hgencoul\Big]
=  \gamma \sum_{j=1}^N \left(  \Big\{L_{ij},  \Big[\frac1r,\nabla_j\Big]\Big\}
- \Big\{\Big[\frac{x_i}{r},\nabla_j\Big],\nabla_j\Big\}\right)
- \gamma \Big[\Big[\frac{1}{r},\nabla_i\Big],S\Big].
\end{equation}
Notice that
\begin{equation}\label{rpi}
\Big[\frac1{r}, \nabla_j\Big] =  \frac{x_j}{r^3}.
\end{equation}
Equalities \eqref{final},  \eqref{rpi} together with \eqref{com-nn} imply that
\be\label{fin2}
\Big[ A_i\,, \, \Hgencoul\Big]
=  \frac{\gamma}{r^3}  \sum_{j=1}^N \left\{L_{ij},x_j\right\}
 - {\gamma}  \sum_{j=1}^N\Big\{ \frac{x_i x_j}{r^3}-\frac{S_{ij}}{r},\nabla_j\Big\}
  + \frac{\gamma}{r^3}[S,x_i].
\ee
Notice that
\be\label{interm}
\frac{1}{r^3}\sum_{j=1}^N \{L_{ij},x_j\}=
\sum_{j=1}^N\left((\nabla_j x_i - \nabla_i x_j)\frac{x_j}{r^3} +  \frac{x_j}{r^3}(x_i \nabla_j - x_j \nabla_i) \right)
= \sum_{j=1}^N\Big\{ \frac{x_ix_j}{r^3},\nabla_j\Big\}-\Big\{\frac1r,\nabla_i\Big\}.
\ee
By substituting \eqref{interm} into the right-hand side of \eqref{fin2} we get
\begin{equation*}
 [A_i,\Hgencoul]
 = {\gamma}  \sum_{j=1}^N\Big\{\frac{S_{ij}-\delta_{ij}}{r},\nabla_j\Big\}
  + \frac{\gamma}{r^3}[S,x_i].
 \end{equation*}
By making use of \eqref{rpi} we obtain
\begin{equation*}
 [A_i,\Hgencoul] =  \frac{\gamma}{r}  \sum_{j=1}^N\Big\{S_{ij}-\delta_{ij},\nabla_j\Big\}
 + \frac{\gamma}{r^3} \Big(x_i-\sum_{j=1}^N x_jS_{ij}+[S,x_i]\Big).
 \end{equation*}
It follows from  formulas  \eqref{xjSij-a}, \eqref{xjSij-b}  that  $[A_i,\Hgencoul] = 0$.
\end{proof}

Other forms of Dunkl LRL vector are given by the following statement.

\begin{prop}\label{compAprop}
Components \eqref{Ai} of the Dunkl LRL vector  $A$  can be represented in the following ways:
\begin{align}
\label{Ai-2}
A_i&=- x_i \Big(\nabla^2+\frac\gamma r\Big)+\nabla_i\Big( r\partial_r+ \sfrac {N-3}2\Big),
\\
\label{Ai-2+}
A_i&=-\Big(\nabla^2+\frac\gamma r\Big)x_i+\Big( r\partial_r+ \sfrac {N+3}2\Big)\nabla_i
\end{align}
for any $i=1, \ldots, N$.
\end{prop}
\begin{proof}
By applying relations \eqref{comLu}, \eqref{xjSij-b} and  \eqref{Siisum} we obtain
the following formula:
\begin{equation*}
\sum_{j=1}^N [\nabla_j,L_{ij}]=-(N-1)\nabla_i + \{ S,\nabla_i\}.
 \end{equation*}
This transfers   expression \eqref{Ai} for the components of Dunkl LRL vector  to the form
\begin{equation*}
 A_i
= - \sum_{j=1}^N  L_{ij}\nabla_j-  \left(S-\sfrac{N-1}{2}\right)\nabla_i -\frac{\gamma x_i}{r}.
 \end{equation*}
By substituting  \eqref{Mij} into the above expression and using the first  identity in \eqref{x-pi} together with the
commutator
 \begin{equation*}
[r\partial_r, \nabla_i]=- \nabla_i,
 \end{equation*}
we arrive
at the first required relation \eqref{Ai-2}.
Then  relation  \eqref{Ai-2+} follows since
 \begin{equation*}
[x_i,\nabla^2]= - 2  \nabla_i
 \end{equation*}
due to \eqref{xjSij-b2}.
Note that one has
\begin{equation*}
(r\partial_r)^+ = - r\partial_r -   N,
 \end{equation*}
which implies formula \eqref{Ai-2} since Dunkl LRL vector satisfies Hermitian property $A_i^+ = A_i$ for any $i$.
%
\end{proof}

Proposition \ref{compAprop} allows to derive commutation relations between components of Dunkl LRL  vector and Dunkl angular momentum.
\begin{lem}\label{newlem}
The following relation holds
\begin{equation*}
[A_i, L_{kl}]= A_{l}S_{ki} - A_{k}S_{li}
\end{equation*}
for any $i,k,l=1, \ldots, N$.
\end{lem}
Lemma \ref{newlem} follows immediately from Lemma \ref{lem2}, Proposition   \ref{compAprop} and formula  \eqref{com-p2L}.


It appears that commutators of components of Dunkl LRL  vector can be expressed in a compact form.
\begin{lem}
\label{newlem2}
The following relation holds
\begin{equation*}
[A_i,A_j]= \Hgencoul L_{ij}
\end{equation*}
for any $i,j=1, \ldots, N$.
\end{lem}
\begin{proof}
Let us consider the product $A_iA_j$, where the operator $A_i$ has the form \eqref{Ai-2+} and the operator $A_j$ has  the form \eqref{Ai-2}.
Let us subtract  from this product the same  expression with the swapped indexes, $i\leftrightarrow j$. We obtain
\be
\label{aiaj54}
[A_i,A_j]=\left(\left ( r\partial_r +  \sfrac {N+3}2\right)\left(\nabla^2+\frac\gamma r\right)
- \left(\nabla^2+\frac\gamma r\right) \left( r\partial_r+ \sfrac {N-3}2\right)\right) L_{ij}
\ee
as $L_{ij}$ commutes with both $r$ and $\nabla^2$ by the proof of Proposition \ref{prop32}.
Since
\be
\label{commagain}
[r\partial_r, \nabla^2]= - 2 \nabla^2,
\qquad
[r \partial_r, r^{-1}]= - r^{-1},
\ee
we get from  \eqref{aiaj54}  that
$$
[A_i, A_j] =  \left(\frac{2\gamma}{r}+\nabla^2\right)L_{ij}= \Hgencoul L_{ij}
$$
as required.
\end{proof}
The next lemma gives a compact expression for the squared length of the LRL vector.
\begin{lem}
\label{newlem3}
Let $A^2 =\sum_{i=1}^N A_i^2$. Then
the following relation holds
\be
\label{Asq}
A^2
=   \Hgencoul\left(\mathcal{I}+S-\sfrac{(N-1)^2}{4}\right) + \gamma^2,
\ee
where  $\mathcal{I}$ is given by formula \eqref{Lsq}.
\end{lem}
\begin{proof}
Formulas \eqref{Ai-2}, \eqref{Ai-2+} together with the relation \eqref{x-pi} imply that
\be
\label{Asq2}
\begin{split}
A^2
=&\Big(\nabla^2+\frac\gamma r\Big)\Big(r^2\nabla^2+\gamma r\Big)+\Big( r\partial_r+ \sfrac {N+3}2\Big)\nabla^2\Big( r\partial_r+\sfrac {N-3}2\Big)
\\
&-\Big(\nabla^2+\frac\gamma r\Big)(r\partial_r+S)\Big( r\partial_r+\sfrac {N-3}2\Big)
\\
&-\Big( r \partial_r+ \sfrac {N+3}2\Big)(r\partial_r-S+N))\Big(\nabla^2+\frac\gamma r\Big).
\end{split}
\ee
Taking into account the first commutator in \eqref{commagain},
we can simplify \eqref{Asq2} to the form
\be
\label{quadrgam}
\begin{split}
A^2
=&\Big(\nabla^2+\frac\gamma r\Big)\Big(r^2\nabla^2+\gamma r\Big)+ \left( \sfrac{N-1}{2}\nabla^2
- \gamma \partial_r -\left(\nabla^2+\frac{\gamma}{r}\right)  S\right)\left( r\partial_r+ \sfrac {N-3}2\right)
\\
&- \left( r\partial_r+ \sfrac {N+3}2\right)(r\partial_r - S+N)\Big(\nabla^2 + \frac\gamma r\Big).
\end{split}
\ee
The right-hand side of the equality \eqref{quadrgam} is a second order polynomial in $\gamma$. Let us consider firstly terms which do not contain $\gamma$. They have the form
\be
\label{rearag}
\begin{split}
&\nabla^2\left(r^2 \nabla^2+\left( \sfrac{N-1}{2}- S\right)(r\partial_r+ \sfrac {N-3}2) \right) -
\left( r\partial_r+ \sfrac {N+3}2\right)(r\partial_r-S+N)) \nabla^2
\\
&=
\nabla^2\left(r^2 \nabla^2+\left( \sfrac{N-1}{2}- S\right)\left(r\partial_r+\sfrac {N-3}2\right)  -
\left(r\partial_r+\sfrac{N-1}2\right) (r\partial_r -  S+N-2)\right)
\end{split}
\ee
due to \eqref{commagain}. We rearrange \eqref{rearag}  further to   
\be
\label{icalhelp}
\nabla^2\left(r^2 \nabla^2 -(r \partial_r)^2 -   (N-2) r \partial_r + S -\sfrac{(N-1)^2}{4}\right) =  \nabla^2\left( {\mathcal I} + S - \sfrac{(N-1)^2}{4}\right)
\ee
since
${\mathcal I}=  r^2 \nabla^2 -  r^2 \partial_r^2 - (N-1) r \partial_r$ by formula \eqref{Hzero}, and
 $(r \partial_r)^2 = r^2 \partial_r^2 + r \partial_r$.

Let us now consider terms in \eqref{quadrgam} containing $\gamma$ in power one. We have
\be
\label{lingamma}
\begin{split}
r \nabla^2 + \nabla^2 r - (\partial_r +   r^{-1} S)\left(r \partial_r +\sfrac{N-3}{2}\right) -
\left(r \partial_r +  \sfrac{N+3}{2}\right)(r \partial_r + S + N )r^{-1}
\\
=r^{-1}\left(r^2 \nabla^2 + r \nabla^2  r - 2 (r \partial_r)^2 - 2  (N-1) r \partial_r + 2S - \sfrac{N^2-1}{2}\right),
\end{split}
\ee
where we used
$[\partial_r, r^{-1}]= - r^{-2}$. It follows from the formula  \eqref{Hzero} that
$$
[\nabla^2, r]=\left[\partial^2_r +\sfrac{N-1}{r} \partial_r, r\right] =  2  \partial_r + (N-1)r^{-1}.
$$
Therefore
\be
\label{onemore}
r \nabla^2 r = r^2 \nabla^2 +r [\nabla^2,r] = r^2 \nabla^2 + 2   r \partial_r +N-1.
\ee
Substituting expression  \eqref{onemore} into formula  \eqref{lingamma} we obtain the expression
\begin{equation*}
r^{-1}\left( 2 r^2 \nabla^2 - 2 (r \partial_r)^2 - 2  (N-2) r \partial_r + 2S -\sfrac{(N-1)^2}{2}\right)
=  2 r^{-1}\left({\mathcal I} + S - \sfrac{(N-1)^2}{4}\right)
\end{equation*}
by the relation \eqref{icalhelp}. The statement follows.
\end{proof}

The next statement generalises orthogonality relation between angular momenta and components of LRL   vector in three-dimensional space.

\begin{lem}
\label{orthogamRL}
For any $i,j,k=1, \ldots, N$ we have
$$
L_{ij}A_k + L_{jk} A_i + L_{ki} A_j = A_k L_{ij} + A_i L_{jk}+ A_j L_{ki} = 0.
$$
\end{lem}
\begin{proof}
The statement follows from Lemma \ref{orthogrel}  and  commutativity
$$
[\nabla^2, L_{ij}]=[r^{-1}, L_{ij}]=[r \partial_r, L_{ij}]=0.
$$
\end{proof}

\begin{rem}\label{remark1}
Commutativity property of the non-local Hamiltonian $\Hgencoul$ established in Propositions \ref{prop32}, \ref{prop33} allows to obtain integrals of the local Hamiltonian \cite{CalCoul}
\be
\label{gcoulloc}
H_{\gamma, g}^{loc} =
\Delta
-  \sum_{\alpha\in {\mathcal R}_+}\frac{g_\alpha(g_\alpha-1) (\alpha, \alpha)}{x_\alpha^2} +  \frac{2\gamma}{r},
\ee
similarly to the case $\gamma=0$ \cite{Heckman}. Namely, let $\rm{Res}(B)$ be restriction of a $W$-invariant  operator  $B$ to $W$-invariant functions. Then  $H_{\gamma, g}^{loc} = \rm{Res} (\Hgencoul)$. Furthermore, if 
$P$ is a polynomial in the non-commuting variables  $A_i, L_{kl}$ and in elements $w\in W$ such that $P$ is $W$-invariant then it follows that $[H_{\gamma, g}^{loc}, {\rm Res}(P)]=0$. Note that integrals corresponding to two $W$-invariant polynomials $P_1, P_2$ do not commute one with another in general.
\end{rem}

\begin{rem}
One may also adopt approach of \cite{Heckman}, \cite{feigin} to construct shift operators for the operator \eqref{gcoulloc}, namely, differential operators $\mathcal D$ such that
\be
\label{shift}
H_{\gamma, g+1}^{loc} \circ {\mathcal D}=  {\mathcal D}\circ H_{\gamma, g}^{loc}.
\ee
These operators arise from the application of $W$-anti-invariant operators $P$  to $W$-invariants. In the case of $P$ not dependent on variables $A_i$ such shift operators of degree 0 were constructed in \cite{feigin} for $\gamma=0$, and they satisfy intertwining relation \eqref{shift} for any $\gamma$.
\end{rem}

\section{Algebraic structure}
\label{section-BR}

Let us introduce the following associative algebra $R_{g,\gamma}(W)$ over $\C$ with identity $1\in W$. It is generated by the elements $\mathcal A_i$, $\mathcal L_{kl}$, $\mathcal H$ and $\C W$, where $1\le i,k,l \le N$. It is convenient to define
$$
\mathcal A_\xi = \sum_{i=1}^N \xi_i {\mathcal A}_i,  \quad \mathcal L_{\xi \eta} =\sum_{i,j=1}^N \xi_i \eta_j \mathcal L_{ij},
$$
where $\xi=(\xi_1, \ldots, \xi_N), \eta = (\eta_1, \ldots, \eta_N) \in \C^N$. Then defining relations of the algebra $R_{g,\gamma}(W)$ are
\begin{align}
\label{rep1}
&\mathcal L_{i j} = - \mathcal L_{j i}, \\
\label{rep2}
& w \mathcal A_\xi = \A_{w(\xi)} w, \quad w \mathcal L_{\xi \eta} = \mathcal L_{w(\xi) w(\eta)} w,  \\
\label{rep3extra}
& [\mathcal H, \A_i] = [\mathcal H, \mL_{ij}]=0, \\
\label{rep3}
& \sum_{i=1}^N \A_i^2=  \mathcal H \left(\sum_{i<j}^N \mathcal L_{ij}^2 - S(S-N+1) - \sfrac{(N-1)^2}{4}\right) + \gamma^2,\\
\label{rep4}
& [\A_i, \A_j] =   \mathcal H \mathcal L_{ij}, \\
\label{rep5}
& [\A_i, \mathcal L_{kl}]= \A_{l}S_{ki} -  \A_{k}S_{li}, \\
\label{rep6}
& [\mathcal L_{ij}, \mathcal L_{kl}] =  \mathcal L_{il} S_{jk}  + \mathcal L_{jk} S_{il} - \mathcal L_{ik} S_{jl} -  \mathcal L_{jl} S_{ik},\\
\label{rep7}
& \mL_{ij} \mL_{kl} + \mL_{jk} \mL_{il}+\mL_{ki}\mL_{jl} =   \mL_{ij} S_{kl} + \mL_{jk} S_{il} + \mL_{ki} S_{jl}, \\
\label{rep8}
& \mL_{ij} \A_{k} + \mL_{jk} \A_i+\mL_{ki} \A_j =0,
\end{align}
where $i,j,k,l=1, \ldots, N$. Note that algebras $R_{g,\gamma}(W)$ are isomorphic for all $\gamma\ne 0$ as one can rescale generators
$\A_{i} \to \gamma \A_{i}$, $\mathcal H \to \gamma^2 \mathcal H$.

\begin{rem}
Let $g=0$. Then $S_{ij}=\delta_{ij}$.  Let us also set ${\mathcal H}=c \in \C^\times$. Consider universal enveloping algebra $U(so(N+1))$ with generators $K_{ij} = E_{ij} - E_{ji}$ where $E_{ij}$ is the matrix unit at $(ij)$-place, $1\le i < j\le N+1$.
For  $\xi=(\xi_1,\ldots, \xi_{N+1}), \eta=(\eta_1,\ldots, \eta_{N+1})$ define $E_{\xi \eta} = \sum_{i,j=1}^{N+1} \xi_i \eta_j E_{ij}$. Consider the skew product $U(so(N+1))  \rtimes  W$ where group $W$ acts naturally in $\C^N$
embedded into the subspace of $\C^{N+1}$
with the last coordinate of elements being 0, $w E_{\xi \eta} = E_{w(\xi), w(\eta)}w$. Due to relations \eqref{rep4} -- \eqref{rep6}  we have a surjective homomorphism $\psi\colon U(so(N+1))\rtimes W \to R_{g, \gamma}(W)/({\mathcal H} - c)$ given by $\psi(K_{ij}) = {\mathcal L}_{ij}$ for $1\le i<j \le N$, 
$\psi(K_{i, N+1})  = (-c)^{1/2} {\mathcal A}_i$  ($1\le i \le N$), $\psi(w) = w$ ($w \in W$). Relations \eqref{rep3}, \eqref{rep7}, \eqref{rep8} then define ideal $I = \text{ Ker } \psi \subset  U(so(N+1))\rtimes W$ so that $U(so(N+1))\rtimes W/ I \cong R_{g, \gamma}(W)/({\mathcal H} - c)$.
\end{rem}

Let us consider the algebra $D$ of operators on functions which is generated by differential operators with rational coefficients, function $r^{-1}$ and the group $W$.
\begin{prop}
\label{rhohom}
There exists a homomorphism $\rho\colon R_{g,\gamma}(W) \to D$ given by
\be
\label{rho}
  \rho(\mL_{ij}) = L_{ij},
  \quad \rho(\A_{i}) = A_i,
  \quad \rho({\mathcal H}) = {\mathcal H}_\gamma,
\quad
  \rho(w) = w,
\ee
where $w\in W$, $i,j=1,\ldots,N$.
\end{prop}
\begin{proof}
One has to establish relations \eqref{rep1} - \eqref{rep8} with $\mL_{ij}$ replaced with $L_{ij}$,  $\A_i$ replaced with $A_i$, and $\cal H$ replaced with ${\cal H}_\gamma$. The corresponding relations \eqref{rep1} - \eqref{rep5} follow from the definitions of $A_i, L_{ij}$ and ${\cal H}_\gamma$, Propositions \ref{prop32}, \ref{prop33} and Lemmas \ref{newlem}, \ref{newlem2}, \ref{newlem3}. The corresponding relations \eqref{rep6}, \eqref{rep7} follow from \eqref{comLL}, \eqref{cros}. Finally, the relation \eqref{rep8}
follows from  Lemma \ref{orthogamRL}.
\end{proof}
Following \cite{fh} we will call relations \eqref{rep7}, \eqref{rep8} {\it crossing} relations. This terminology has the following background. Let us consider an element from the algebra $R_{g,\gamma}(W)$ of the form
\begin{equation}
\label{monom}
\mL_{i_1 j_1}^{n_1}\ldots \mL_{i_k j_k}^{n_k},
\qquad
\end{equation}
where $k \in \N$, $1\le i_s<j_s \le N$ for $s=1, \ldots, k$, $n_s \in  \N$, and we assume that pairs of indecies $(i_s, j_s)$ are all different.
Let us plot the integral points from $1$ to $N$ on the real line
 and let us represent element \eqref{monom} geometrically
by connecting the points $i_s$ and $j_s$ by $n_s$
arcs in the upper half-plane.
We say that element
\eqref{monom}
has {\it no crossings} if the corresponding arcs  do not intersect.
This property  can also be restated as follows. Suppose that $i_s < i_{s'} < j_s$ for some $1\le s, s' \le k$. Then this implies  that $j_{s'} \le j_s$.
Now let us take four indecies $1\le i<j<k<l\le N$. Then exactly two of the three elements  $\mL_{ij} \mL_{kl}$,  $\mL_{jk} \mL_{il}$, $\mL_{ik}\mL_{jl}$, namely, elements $\mL_{ij} \mL_{kl}, \mL_{jk} \mL_{il}$   have no crossings.

Furthermore, similarly to elements of the form \eqref{monom} let us consider elements from $R_{g, \gamma}(W)$ of the form
\begin{equation}
\label{monom-with-A}
\mL_{i_1 j_1}^{n_1}\ldots \mL_{i_k j_k}^{n_k} \A_{r_1}^{m_1}\ldots\A_{r_l}^{m_l},
\qquad
\end{equation}
where $i_s, j_s, n_s$ are as in \eqref{monom}, $l\in \N$, $1\le r_t\le N$ and $m_t\in \N$ for $t=1, \ldots, l$ ($r_t\ne r_{t'}$ if $t\ne t'$).
Let us represent this element geometrically as follows. Let us plot the integral points on the real line from $1$ to $N+1$. Let us connect pairs of points $(i_s, j_s)$ by $n_{s}$
arcs in the upper half-plane as we did for the element \eqref{monom}.
Let us also  connect pairs of points $(r_t, N+1)$ by $m_t$ arcs in the upper half-plane ($t=1, \ldots, l$).  We say that element \eqref{monom-with-A}  has {\it no crossings} if the corresponding arcs do not intersect.
Then for any triple of  indecies $1\le i<j<k\le N$ exactly two of the three elements  $\mL_{ij} \A_{k}$,  $\mL_{jk} \A_i$, $\mL_{ik} \A_j$ in \eqref{rep8}, namely, elements  $\mL_{ij} \A_{k}$,  $\mL_{jk} \A_i$  have no crossings.

More generally, we consider elements of $R_{g, \gamma}(W)$ of the form \eqref{monom-with-A} multiplied (on the right) by ${\cal H}^q w$, where $q\ge 0$ and $w \in W$. We say that such an element has no crossings if the corresponding element \eqref{monom-with-A} has no crossings, and we represent such an element geometrically in the same way as element \eqref{monom-with-A} for any $q, w$.

Similarly, an element
$L_{i_1 j_1}^{n_1}\ldots L_{i_k j_k}^{n_k} A_{r_1}^{m_1}\ldots A_{r_l}^{m_l} {\cal H}^q w\in D$ is said to have no crossings
($m_t, q\ge 0$)
if the corresponding element
\begin{equation}
\label{standardform}
\mL_{i_1 j_1}^{n_1}\ldots \mL_{i_k j_k}^{n_k} \A_{r_1}^{m_1}\ldots\A_{r_l}^{m_l} {\cal H}^q w
\end{equation}
in  $R_{g, \gamma}(W)$ has no crossings.

It follows from the relations \eqref{rep7}, \eqref{rep8} that any quadratic polynomial in $\mL_{ij}, \A_r$ can be rearranged as a linear combination of elements from $R_{g, \gamma}(W)$ with  no crossings. Furthermore, any element of $R_{g, \gamma}(W)$  can be represented as a linear combination of elements of the form \eqref{standardform}  with no crossings.


We will also use similar terminology for the classical version of above considerations. More exactly, let us consider classical angular momenta $M_{ij} = x_i p_j - x_j p_i$, where $1\le i, j \le N$, and all the variables $x_1, p_1, \ldots, x_N, p_N$ commute. Consider monomial
\begin{equation}
\label{monom-class}
M_{i_1 j_1}^{n_1}\ldots M_{i_k j_k}^{n_k},
\end{equation}
where as in \eqref{monom} $k \in \N$, $1\le i_s<j_s \le N$ for $s=1, \ldots, k$, $n_s \in  \N$, and we assume that pairs of indecies $(i_s, j_s)$ are all different.
We say that element  \eqref{monom-class} has no crossings if the corresponding element \eqref{monom} has no crossings. Equivalently, 
 the arcs constructed in the same way as for the monomial \eqref{monom} do not intersect.
\begin{lem}
Any monomial \eqref{monom-class} 
can be represented as a linear combination of monomials in the classical angular momenta which have no crossings.
\end{lem}
\begin{proof}
Let us consider a monomial which contains a factor $M_{ij}M_{kl}$ where $i<k< j<l$. It is easy to see that by applying the crossing relation
$$
M_{ij}M_{kl} = M_{ik} M_{jl} - M_{il}M_{jk}
$$
we obtain two monomials such that each of them has less intersecting pairs of arcs than the original monomial. The statement follows by induction.
\end{proof}

\begin{lem}\label{mind}
Different
monomials \eqref{monom-class} 
are linearly independent provided that they have  no crossings.
\end{lem}
\begin{proof}
Let $i=(i_1, \ldots, i_k)$,  $j=(j_1, \ldots, j_k)$,  $n= (n_1, \ldots, n_k)$.
Suppose that monomials \eqref{monom-class} with no crossings admit linear dependence
\begin{equation}
\label{relnl}
\sum_{i,j,n; k} a_{i j n}^{(k)} M_{i_1 j_1}^{n_1} \ldots M_{i_k j_k}^{n_k} =0,
\end{equation}
where $a_{i j n}^{(k)}\in \C^\times$. We can assume that not all the monomials in \eqref{relnl} are divisible by $M_{12}$ as otherwise we can divide this polynomial relation by a suitable power of $M_{12}$. Furthermore, note that each 
angular momentum  
$M_{i_r j_r}$ is homogeneous in any pair of variables $(x_s, p_s)$. Therefore we can assume that all monomials in the relation (\ref{relnl}) have the same multidegree 
$(k_1, \ldots, k_N)$ 
in the pairs of variables $(x_{s}, p_{s})$ 
($s=1,\ldots, N$). Consider now the homomorphism $\varphi$
from the algebra generated by angular momenta in $N+1$ variables to the algebra generated by angular momenta in the last $N$ variables
given by $\varphi(M_{1 i}) = M_{2i}$ for any $i\ge 3$, $\varphi(M_{12})=0$, and $\varphi(M_{ij})=M_{ij}$ if $i,j \ge 3$. The relation (\ref{relnl})  implies that
\begin{equation}
\label{relnlphi}
\sum_{i,j,n; k} a_{i j n}^{(k)} \varphi(M_{i_1 j_1})^{n_1} \ldots \varphi(M_{i_k j_k})^{n_k} =0,
\end{equation}
and each non-zero monomial $\varphi(M_{i_1j_1})^{n_1}\ldots \varphi(M_{i_k j_k})^{n_k}$ has no crossings. Furthermore, given a nonzero monomial  $\overline{M}= \varphi(M_{i_1j_1})^{n_1}\ldots \varphi(M_{i_k j_k})^{n_k}$ in \eqref{relnlphi} there exists a unique monomial $M$ in \eqref{relnl} such that $\varphi(M)=\overline{M}$ (namely, $M= M_{i_1j_1}^{n_1}\ldots M_{i_k j_k}^{n_k}$).
Indeed,  monomial $M$ has no crossings and it has fixed degrees $k_1, k_2$.
The latter condition means geometrically that monomial $M$ has $k_1$ arcs attached to the point 1 and $k_2$ arcs  attached to the point 2, which allows to find monomial $M$ uniquely for a given $\overline{M}$.
It follows by induction in 
$N$ that coefficients at monomials not containing $M_{12}$ in the sum \eqref{relnl} are all zero. This contradiction implies the statement.
\end{proof}

For a given $m\in \N$ let us introduce the following combination of classical angular momenta:
$$
M^2_{(m)} = \sum_{1\le i<j \le m} M_{ij}^2 = \sum_{i=1}^m x_i^2 \sum_{i=1}^m p_i^2 - \left(\sum_{i=1}^m x_i p_i \right)^2.
$$

Consider the algebra $B=B_{N+1} \subset \C[x_1, \ldots, x_{N+1}, p_1, \ldots, p_{N+1}]$ generated by classical angular momenta $M_{ij}$, $1\le i<j\le N+1$.
Let $I=I_{N+1}$ be the ideal in $B$ generated by the element $M^2_{(N+1)}$.

\begin{lem}\label{lemmom}
A linear basis in the quotient
 $B/I$
 is given by the coset classes of different monomials $M_{i_1 j_1}^{n_1} \ldots M_{i_k j_k}^{n_k}$ ($i_s<j_s$ for all $s=1,\ldots, k$) such that they have no crossings and 
$M_{N, N+1}$ has power at most $1$.
\end{lem}
\begin{proof}
Let $\overline{M}_{ij}$ be the image of $M_{ij}$ in $B/I$. Suppose there is a linear dependence between the specified monomials in the quotient, that is  $Q(\overline{M})=0$, where $Q(\overline{M})$ is a linear combination of monomials 
$\overline{M}_{i_1 j_1}^{n_1} \ldots \overline{M}_{i_k j_k}^{n_k}$.
Let $Q(M)$ be the same linear combination of the elements 
 ${M}_{i_1 j_1}^{n_1} \ldots {M}_{i_k j_k}^{n_k}$.
We get the following polynomial relation in the algebra $B$:
\begin{equation}
\label{relB}
Q(M)=M^2_{(N+1)} q(M),
\end{equation}
where $q(M)$ is a polynomial in variables $M_{ij}$. We can assume that monomials in the polynomial $q(M)$ have no crossings. Note that $q\ne 0$ by Lemma  \ref{mind}. 
Let us consider terms $Q_0, q_0$ in the polynomials $Q(M)$ and  $q(M)$ respectively which have maximal total degree in the pairs of variables 
$(x_N, p_N)$ and $(x_{N+1}, p_{N+1})$ (equivalently, the corresponding monomials have maximal possible number of indexes 
$N$ and $N+1$). Then relation \eqref{relB} implies the relation 
$Q_0 = M_{N, N+1}^2 q_0$.
 Monomials in $Q_0$ have no crossings and they contain 
$M_{N, N+1}$ in the power at most 1.
Monomials in 
$M_{N, N+1}^2 q_0$ have no crossings too and they contain  
$M_{N, N+1}^2$ in the power at least 2.
This is a contradiction with Lemma \ref{mind} which implies that monomials
 $M_{i_1 j_1}^{n_1} \ldots M_{i_k j_k}^{n_k}$ are linearly independent.
\end{proof}

\begin{lem}\label{keylem}
Let $x=(x_1,\ldots, x_{N+1}), p = (p_1,\ldots, p_{N+1})\in \C^{N+1}$ be generic vectors satisfying relation $M^2_{(N+1)} =0$.
Then there exist $\lambda, \mu \in \C$ such that
\begin{equation}
\label{anz}
\hat x = (\hat x_1, \ldots, \hat x_{N+1}) = \lambda x,
\qquad
\hat p =(\hat p_1, \ldots, \hat p_{N+1}) =  \lambda^{-1} p + \mu x
\end{equation}
satisfy
\begin{equation}
\label{xnewpnew}
\sum_{k=1}^{N+1} \hat x_k \hat p_k =0,
\qquad
\sum_{k=1}^{N+1} \hat p_k^2 =0.
\end{equation}
One also has $M_{ij} = \hat x_i \hat p_j - \hat x_j \hat p_i$ for any $i,j = 1, \ldots, N+1$.
\end{lem}
\begin{proof}
By substituting expressions \eqref{anz} into \eqref{xnewpnew} we get conditions
\begin{align}
\label{cond1}
&\sum_{k=1}^{N+1}   x_k  p_k + \lambda \mu \sum_{k=1}^{N+1} x_k^2 =0,
\\
\label{cond2}
&\lambda^{-2}\sum_{k=1}^{N+1}  p_k^2 + 2 \lambda^{-1}\mu  \sum_{k=1}^{N+1}  x_k  p_k +  \mu^2 \sum_{k=1}^{N+1}  x_k^2 =0.
\end{align}
The condition \eqref{cond2} can then be replaced with the relation 
\begin{equation}
\label{cond3}
\lambda^2 \mu^2 \sum_{k=1}^{N+1}  x_k^2 = \sum_{k=1}^{N+1}  p_k^2.
\end{equation}
Note that relation  \eqref{cond3} follows from \eqref{cond1} since
$$
M_{(N+1)}^2=
\sum_{k=1}^{N+1}   x_k^2 \sum_{k=1}^{N+1}   p_k^2 - \left(\sum_{k=1}^{N+1}  x_k  p_k \right)^2 =
0.
$$
We get that
$$
\lambda \mu = - \frac{ \sum_{k=1}^{N+1} x_k p_k}{\sum_{k=1}^{N+1} x_k^2},
$$
which determines required $\lambda$ and $\mu$ up to a scaling factor.
\end{proof}

We are going to use previous lemmas in the study of baseses in the non-commutative algebras $R_{g, \gamma}(W)$, $\rho(R_{g, \gamma}(W))$. It is convenient to fix an ordering when writing elements of these algebras. It is clear that any element of the algebra $R_{g, \gamma}(W)$ can be represented as a linear combination of elements
\begin{equation}
\label{elb}
{\mathcal L}_{12}^{i_{12}}\ldots {\mathcal L}_{N-1,N}^{i_{N-1,N}}  {\mathcal A}_1^{k_1}\ldots {\mathcal A}_N^{k_N} \mathcal H^l w,
\end{equation}
where we write $\mL_{r s }$ to the left of $\mL_{r' s'}$ 
($1\le r<s\le N, 1\le r'< s'\le N$)   
if $r<r'$ or $r=r'$ and $s<s'$, and
all the powers in \eqref{elb} are natural numbers or zeroes. We will use the same ordering for the elements of algebra
$\rho(R_{g, \gamma}(W))$ in the next theorem.

\begin{thm}\label{PBWth} Consider elements of the algebra $D$ of the form
\begin{equation}
\label{basisin}
L_{12}^{i_{12}}\ldots L_{N-1, N}^{i_{N-1,N}} A_1^{k_1}\ldots A_N^{k_N} {\cal H}^l w
\end{equation}
with $i_{12}, \ldots, i_{N-1,N}, k_1, \ldots k_N, l \in  \N\cup \{0\}$,
$w \in W$ such that they have no crossings and   
$k_N\le 1$.
Then all such elements  are linearly independent.
\end{thm}

\begin{proof} Let us
assume a linear dependence of the elements of the form \eqref{basisin}. Then there is a linear dependence between their  highest symbols. To get the highest symbols one replaces $L_{ij}$   with the classical angular momentum $M_{ij} =x_i p_j - x_j p_i$, one replaces $A_i$   with
$$
A_i^{cl}= - \sum_{j=1}^N p_j(x_i p_j - x_j p_i),
$$
 and one replaces 
${\cal H}_\gamma$ with 
$H^{cl}= \sum _{i=1}^N p_i^2$. 
 It is sufficient to consider the case $w=1$. We can assume that the highest symbol is homogeneous if $\deg p_i=1$ and  $\deg x_i = -1$ 
for all $i$. 
Therefore, after dividing by a suitable power of $(- \sum p_i^2)^{1/2}$ we get a polynomial $Q$ in 
$A_i^{cl}/(-H^{cl})^{1/2}$, $M_{ij}$ ($1\le i<j \le N$) 
which is equal to zero.

Let us introduce new variables
\begin{equation}
\label{newxpnp1}
p_{N+1}=\left(-\sum_{j=1}^N p_j^2\right)^{1/2}, \qquad x_{N+1} = -\frac{\sum_{j=1}^N x_j p_j}{p_{N+1}}.
\end{equation}
Notice that
$$
M_{i, N+1} = x_i p_{N+1} - x_{N+1} p_i = \frac{p_i \sum_{j=1}^N x_j p_j - x_i\sum_{j=1}^{N} p_j^2}{\big(-\sum_{j=1}^N p_j^2\big)^{1/2}}
= \frac{\sum_{j=1}^N p_j (x_j p_i - x_i p_j)}{\big(-\sum_{j=1}^N p_j^2\big)^{1/2}}
 =
\frac{A_i^{cl}}{(-H^{cl})^{1/2}}.
$$
 We get that $Q(M) =0$, where the polynomial $Q(M)$ is a polynomial in variables $M_{ij}$ ($1\le i<j\le N+1$) obtained by replacing variables  
$A_i^{cl}/(-H^{cl})^{1/2}$
 with 
$M_{i, N+1}$ ($1\le i \le N$).
Note that $(x_1, \ldots, x_{N+1})$, $(p_1,\ldots, p_{N+1})$ satisfy the condition $M^2_{(N+1)} =0$. Furthermore, it follows by Lemma \ref{keylem} that for any  generic point on $M^2_{(N+1)}=0$ we can assume that $x_{N+1}, p_{N+1}$ are given by formulas \eqref{newxpnp1} without change of the angular momenta $M_{ij}$, $1\le i<j \le N+1$.
Therefore we get that $Q(M)=0$, where angular momenta satisfy $M^2_{(N+1)}=0$, and there are no other constraints on angular momenta $M_{ij}$. 
By Hilbert's Nullstellensatz we get that $Q(M)^k\in I_{N+1}$ for some $k \in \N$.
Since the scheme $M^2_{(N+1)}=0$ is integral (see e.g.  \cite{FT}), it follows
that $k=1$ so that $Q(\overline{M}) =0$ in the quotient algebra $B_{N+1}/I_{N+1}$, where $Q(\overline{M})$ stands for replacing each variable  $M_{ij}$ in the polynomial $Q$ with its image $\overline{M_{ij}}$ in the quotient.
By Lemma \ref{lemmom} we get that $Q$ is the zero polynomial.
\end{proof}


Theorem \ref{PBWth} allows us to give a linear basis for the algebra $R_{g,\gamma}(W)$.

\begin{thm}\label{PBWR}
Consider elements of the algebra $R_{g,\gamma}(W)$ of the form
\begin{equation}
\label{elR}
{\mathcal L}_{12}^{i_{12}}\ldots {\mathcal L}_{N-1,N}^{i_{N-1,N}}  {\mathcal A}_1^{k_1}\ldots {\mathcal A}_N^{k_N} \mathcal H^l w,
\end{equation}
where $w\in W$, and $i_{12},\ldots, i_{N-1,N}, k_1, \ldots, k_N, l \in \N\cup\{0\}$ are such that 
$k_N\le 1$ and there is   no crossings. These elements form a basis in  the algebra $R_{g,\gamma}(W)$.
\end{thm}
\begin{proof}
Firstly we note that relations of the algebra $R_{g, \gamma}(W)$ imply that any element of the algebra can be represented as a linear combination of the elements of the form \eqref{elR} which have no crossings and where $k_N\in\N\cup\{0\}$ is arbitrary.
 Suppose now that $k_N\ge 2$. Let us apply the relation \eqref{rep3} in the form
\begin{equation}
\label{relA}
\A_{N}^2 = - \sum_{i=1}^{N-1} \A_i^2 + \mathcal H \left( \sum_{i<j}^N \mathcal L_{ij}^2
 - S(S-N+1)- \sfrac{(N-1)^2}{4}\right) + \gamma^2
\end{equation}
to the monomial \eqref{elR}, and let us rearrange the resulting terms as a linear combination of monomials with no crossings. For a given monomial and its geometrical realisation let $P$ be the total number of arcs ending in the points $N$ or $N+1$ with arcs connecting these points counted twice. Then as we apply relation \eqref{relA} the maximum of $P$ across all the resulting monomials is reduced by at least 2 in comparison with the original monomial. By continuing the process we arrive at monomials with $k_N\le 1$. Therefore elements \eqref{elR} with no crossings and 
$k_N\le 1$ span the algebra $R_{g, \gamma}(W)$. It follows from Proposition \ref{rhohom} and Theorem \ref{PBWth} that these elements are linearly independent.
\end{proof}

Theorem \ref{PBWR} can be interpreted as Poincar\'e--Birkhoff--Witt theorem for the algebra $R_{g, \gamma}(W)$.  Theorems \ref{PBWth}, \ref{PBWR}  also have the following implication.

\begin{cor}\label{corind}
Homomorphism  of algebras $\rho\colon R_{g, \gamma}(W) \to D$ given by formulas  \eqref{rho} has trivial kernel.
\end{cor}

Consider the Dunkl angular momenta algebra $H_{g}^{so(N+1)}(W)$, where group $W$ acts by its geometric representation on $V\cong \C^N$. Vector space
$V$ is embedded into $\C^{N+1}$ via the first $N$ components:
$$
V\ni (x_1,\ldots, x_N)\to (x_1, \ldots, x_N, 0)\in \C^{N+1},
$$
 and the group $W$ acts trivially on the last component of vectors in $\C^{N+1}$.
It appears that a central quotient of the algebra $R_{g,\gamma}(W)$ is isomorphic to a central quotient of the corresponding  angular momenta   algebra $H_{g}^{so(N+1)}(W)$.

\begin{thm}
For any $a\in \C^\times$ there is an isomorphism of algebras
$$
\varphi: \quad
R_{g,\gamma}(W)\big/({\mathcal H}-a) \to H_{g}^{so(N+1)}(W)\big/ ({\cal I}_{N+1} - b),
$$
where $b = \frac{(N-1)^2}{4} - \frac{\gamma^2}{a}$. 
Under this isomorphism
$$
\varphi({\cal L}_{ij}) = L_{i j}, \quad
\varphi({\cal A}_{i}) = (-a)^{-1/2} L_{i, N+1}, \quad
\varphi(w)=w,
$$
where $1\le i, j \le N$, $w\in W$.
\end{thm}

\begin{proof}
Note that elements 
$$
S_{ij}   =  \delta_{ij} + \sum_{\alpha\in {\mathcal R}_+} \frac{2 g_\alpha \alpha_i\alpha_j}{(\alpha, \alpha)} s_\alpha
  \in \C W \subset H_g^{so(N+1)}(W),
$$
$1\le i, j \le N+1$, 
 satsify the properties
$$
S_{N+1, N+1}=1, \quad S_{i, N+1}=  S_{N+1,i} = 0  \quad  (1\le i \le N).
$$
Therefore relations \eqref{rep4}, \eqref{rep5}, \eqref{rep6} correspond under the map $\rho$ to all the cases of the  relation \eqref{comLL}  (with $1\le i,j,k,l\le N+1$). Similarly, relations \eqref{rep7}, \eqref{rep8} correspond to all the cases of the relation \eqref{cros}.
Note also that relation \eqref{rep3} at ${\cal H}=a$ corresponds under $\rho$ to the condition 
${\cal I}_{N+1} = \frac{(N-1)^2}{4} -\frac{\gamma^2}{a}$ (see \eqref{Lsq}), and that all the remaining relations match as well.

\end{proof}

Note that for $\gamma=0$ the algebra $R_{g,0}(W)$ is isomorphic to the subalgebra   $\widehat R_g(W)$    of the rational Cherednik algebra $H_g(W)$ given as $\widehat R_g(W) = \rho(R_{g,0}(W))$.
Furthermore, the central quotient $R_{g,0}(W)/({\mathcal H}-a)$ does not depend on non-zero value of $a$ and it is a flat deformation in $g$ of non-homogeneous quadratic Poincar\'e--Birkhoff--Witt algebras in the sense of \cite{BG}. Indeed, let us consider the central quotient 
$$
\varphi(R_{g,0}(W)/({\mathcal H}-a)) = H_{g}^{so(N+1)}(W)\big/ ({\cal I}_{N+1} - \sfrac{(N-1)^2}{4})
$$  
 instead. It follows from Lemma \ref{lemmom} that for any $g$ the linear basis in this  quotient is given by the images of monomials
$$
L_{12}^{i_{12}}\ldots L_{N, N+1}^{i_{N, N+1}} w,
$$
where $w \in W$, $i_{rs}\in \N\cup\{0\}$, 
$i_{N, N+1}\le 1$, and monomials have no crossings. It is also clear that all the defining relations of the algebra $R_{g,0}(W)/({\mathcal H}-a)$ are quadratic.

\section{Maximal superintegrability}
\label{superint}

In this section we discuss application of Dunkl LRL vector to questions on maximal superintegrability. 
Let us recall that a quantum Hamiltonian $H$ in $N$-dimensional space is maximally superintegrable if there exists $2N-1$ differential operators $J_1=H, J_2, \ldots, J_{2N-1}$ such that $J_i$ are algebraically  independent and $[H, J_k]=0$ for any $k, 1\le k \le 2N-1$ (cf. \cite{woj83}).

Let  $H_{\gamma} = H_{\gamma, g}^{loc}$ be given by
$$
H_{\gamma}=
\Delta
-  \sum_{\alpha\in {\mathcal R}_+}\frac{g_\alpha(g_\alpha-1) (\alpha, \alpha)}{x_\alpha^2} +  \frac{2\gamma}{r}.
$$
Note that when acting on $W$-invariants the differential operator
$
H_{\gamma}= \text{Res} {\mathcal H}_{\gamma}
$,
where $\text{Res}$ denotes restriction of a $W$-invariant operator to $W$-invariants (cf. \cite{Heckman}). 

Let us consider the following algebra
$$
R ^{cl}_N= \langle A_i^{cl}, M_{jk}, p^2\colon 1\le i \le N, 1\le j<k \le N \rangle  \subset \C[x_1, \ldots, x_N, p_1, \ldots, p_N],
$$
where $A_i^{cl}= \sum_{j=1}^N p_j(x_j p_i - x_i p_j)$,   $M_{jk} = x_j p_k - x_k p_j$, $p^2 = \sum_{i=1}^N p_i^2$.

\begin{prop}\label{trdeg}
Transcendence degree of the agebra $R^{cl}_N$ is equal to  $2N-1$.
\end{prop}
\begin{proof}
Consider elements $p^2$, $M_{1 i}$, $2 \le i \le N$, $M_{i, i+1}$, $2\le i \le N-1$ and $A_1^{cl}$. Note that any monomial in these elements has no crossings. It follows from  Theorem \ref{PBWth} and its proof that these elements
 are algebraically independent. It is easy to see using crossing relations that any other element $M_{ij}$, $A_i^{cl}$ can be expressed as a rational function of above elements. The statement follows. 
\end{proof}

\begin{thm}\label{maxint}
Hamiltonian $H_\gamma$ is maximally superintegrable. 
\end{thm}
\begin{proof}
The group $W$ acts on the algebra $R^{cl}_N$. Let $(R^{cl}_N)^W$ be the sublgebra of $W$-invariant elements. By Noether's results $(R^{cl}_N)^W$ is a finitely generated domain and the extension $(R^{cl}_N)^W \subset R^{cl}_N$ is integral (see e.g. \cite[Theorem 13.17]{Eisenbud}). It follows by Proposition \ref{trdeg} that transcendence degree of the algebra $(R^{cl}_N)^W$ is equal to $2N-1$. Let $Q_i$, $1\le i \le 2N-1$ be a homogeneous basis of $(R^{cl}_N)^W$. Let us consider the corresponding elements $Q_i^{q}\in R_{g, \gamma}(W)$ obtained by taking Weyl quantisation of $Q_i$, that is given a monomial in $Q_i$ one gets the corresponding sum of monomials in $Q_i^{q}$ by replacing $p^2 \to {\mathcal H}$, $M_{kl} \to {\mathcal L}_{kl}$, $A_i^{cl} \to {\mathcal A}_i$ and averaging over all possible orderings of variables entering the original monomial. It follows that $Q_i^{q}$ is $W$-invariant. Note that the highest order term of $\rho(Q_i^{q})\in D$ is equal to $Q_i$, where homomorphism $\rho$ is defined in Proposition \ref{rhohom}. Therefore elements $\rho(Q_i^{q})$ are algebraically independent. Define differential operators $J_i = \text{Res} \,  \rho(Q_i^{q})$. It is clear that $[H_{\gamma}, J_i]=0$, and differential operators $J_i$ are algebraically independent. Theorem follows.
\end{proof}

\begin{rem}
In the particular case $\gamma=0$ a large set of quantum integrals for $H_0$ was given in \cite[Proposition II.1]{CFS} and superintegrability was claimed although details on algebraic independence of quantum integrals were omitted. 
\end{rem}

Dunkl LRL vector can also be used to establish maximal superintegrability of certain generalisations of Calogero-Moser systems with Coulomb potential which do not have full Coxeter symmetry. At $\gamma=0$ such systems appeared in the works \cite{CFV}, \cite{SV}. At special values of parameters the corresponding Hamiltonians can be obtained via special restrictions of Dunkl operators \cite{Fsel}. Let us recall these settings in  more detail.

Pairs $(\Pi, g)$ where $\Pi=\Pi_V$ is an intersection of some Coxeter reflection hyperplanes in $V$ and $g$ is a special coupling  parameter were determined in \cite[Theorem 3]{Fsel} so that ideal of functions vanishing on the Coxeter group orbit $W(\Pi)$ is invariant under the corresponding rational Cherednik algebra.  This led to quantum integrability of the Hamiltonian
$$
H_\Pi=\Delta_\Pi
-  \sum_{\alpha\in { \mathcal R}_\Pi } \frac{\widehat g_\alpha (\alpha, \alpha)}{(\alpha, x)^2},
$$
where ${ \mathcal R}_\Pi$ consists of orthogonal projections of vectors from $\mathcal R$ to $\Pi$ which are non-zero, $\Delta_\Pi$ is the Laplace operator on $\Pi$, $x \in \Pi$, and $\widehat g_\alpha$ are coupling parameters determined in terms of $g$ \cite[Theorem 5, Proposition 2]{Fsel}. Quantum integrals were obtained by a restriction on $\Pi$ procedure applied to any $W$-invariant combination of Dunkl operators in a suitable gauge. 

In the proof of the next theorem we apply such a restriction to $W$-invariant combinations $Q$ of Dunkl operators which may depend on $x$-variables as well. The restriction works in the same way as in \cite{Fsel}. In particular, the highest term of the resulting quantum integral is obtained by restriction of the highest term of $Q$ on $\Pi_V\times \Pi_{V^*}$, where $\Pi_{V^*}$ is the image of the plane $\Pi_V$ in the momentum space under the natural isomorphism which identifies $(x_1, \ldots, x_N)$ with $(p_1, \ldots, p_N)$.

 
Assuming that pairs $(\Pi, g)$ are as described we have the following statement on maximal superintegrability of the corresponding Hamiltonians $H_\Pi$ as well as these Hamiltonians with extra Coulomb potential.

\begin{thm}\label{maxgen}
The operator $H_{\Pi, \gamma} = H_\Pi +\frac{2 \gamma}{r}$ is maximally superintegrable for any $\gamma \in \C$, where $r=(x,x)^{1/2}$, $x \in \Pi$. 
\end{thm}
\begin{proof}
Let us consider homomorphism $\pi$ from the ring of polynomials $\C[x_1,\ldots, x_N, p_1,\ldots, p_N]$ to the ring of polynomials on $\Pi_V\times \Pi_{V^*}$ given by natural restriction. Then $\pi(R^{cl}_N)$ is isomorphic to the algebra $R^{cl}_d$, where $d=\dim \Pi$, which can be seen by taking coordinates in such a way that plane $\Pi=\Pi_V$ has equations   $x_N=\ldots=x_{d+1}=0$, and similarly in $p$-space. Therefore transcendence degree of $\pi(R^{cl}_N)$ is equal to $2d-1$ by Proposition \ref{trdeg}. Note that the integral extension $(R^{cl}_N)^W \subset R^{cl}_N$ leads to the integral extension $\pi ((R^{cl}_N)^W) \subset \pi(R^{cl}_N)$, hence transcendence degree of $\pi ((R^{cl}_N)^W)$ is equal to $2d-1$.
For any element $\bar Q\in \pi ((R^{cl}_N)^W)$ let us choose $Q\in (R^{cl}_N)^W$ such that $\pi(Q) =\bar Q$. Consider the quantisation $Q^q$ of $Q$ as in the proof of Theorem \ref{maxint}. Differential operator $\text{Res} \,  \rho(Q^{q})$ commutes with the Hamiltonian $H_\gamma$, and it leads to the quantum integral for $H_{\Pi, \gamma}$ with the highest term $\bar Q$ by the results from \cite{Fsel}. The statement follows. 
\end{proof}

As an example Theorem  \ref{maxgen} establishes maximal superintegrability of the following Hamiltonian describing generalization of Calogero-Moser system to the case of two sets of particles of any sizes $n,m \in\N$ with coordinates $(x_1, \ldots, x_n)$ and $(y_1, \ldots, y_m)$:
$$
H_{n,m} = \sum_{i=1}^{n} \frac{\partial^2}{\partial x_i^2} +  \sum_{j=1}^{m} \frac{\partial^2}{\partial y_j^2} - \sum_{i_1<i_2}^n \frac{2 k (k+1)}{(x_{i_1}-x_{i_2})^2}
- \sum_{j_1<j_2}^m \frac{2 k^{-1} (k^{-1}+1)}{(y_{j_1}-y_{j_2})^2}
-\sum_{i=1}^n \sum_{j=1}^m \frac{2 (k+1)}{(x_{i}-\sqrt{k}y_j)^2} +\frac{2\gamma}{r},
$$
where $r= (\sum_{i=1}^n x_i^2 + \sum_{j=1}^m y_j^2)^{1/2}$, $\gamma \in \C$ and $k \in \N$. Integrability of the Hamiltonian $H_{n,m}$ at $\gamma=0$ and arbitrary $k\in \C^\times$ was established in \cite{CFV} for $m=1$ or $n=1$, and in \cite{SV} in general.

As another example Theorem  \ref{maxgen} establishes maximal superintegrability of the following three-dimensional Hamiltonian:
$$
H=\Delta  
-  \sum_{i=1}^{3} \frac{(4c+1)(4c+3)}{4 x_i^2} -  
 \sum_{\nad{i<j}{\varepsilon_1\in \{1, -1\}}}^{3} \frac{2 c(c+1)}{(x_i +\varepsilon_1 x_j)^2} -   \sum_{\varepsilon_{2,3}\in \{1, -1\} } \frac{6}{(x_1 + \varepsilon_2 x_2 +  \varepsilon_3 x_3)^2} +\frac{2\gamma}{(\sum_{i=1}^3 x_i^2)^{1/2}}, 
$$
where $\Delta = \sum_{i=1}^{3} \frac{\partial^2}{\partial x_i^2}$ and $c, \gamma \in \C$. Integrability of this Hamiltonian was established in \cite{Fsel} for $\gamma=0$.

We refer to \cite{Fsel} for further examples.

\section{Concluding remarks}

In the case of Coulomb problem LRL vector allows to derive spectrum \cite{Pauli} and study scattering \cite{Z}. It would be interesting to try to extend this analysis to Hamiltonians considered in this paper. Development of representation theory of algebras $R_{g, \gamma}(W), H_g^{so(N+1)}(W)$ may be needed which is an interesting direction on its own.

Another open question is on maximal superintegrability of generalised Calogero--Moser systems without full Coxeter symmetry such as Hamiltonian $H_{n,m}$ in the case of general coupling parameter(s), that is parameter $k$ in the case of $H_{n,m}$.

\section*{Acknowledgments}
T.H. is grateful to A. Nersessian for preliminary collaboration and stimulating discussions. 
We are grateful to S. Dubovsky, G. Felder and A.P. Veselov for useful comments. 
The work of T.H. was supported by the Armenian State Committee
of Science Grants No.  18T-1C106, 18RF-002, SFU-02
and was fulfilled within the ICTP Affiliated Center Program AF-04.
\color{black}

\end{document}